\documentclass[submission,copyright,creativecommons]{eptcs}
\providecommand{\event}{GandALF 2020} 
\usepackage{breakurl}             
\usepackage{underscore}           

\usepackage{cite}
\usepackage{mathtools}
\usepackage{xspace}
\usepackage{amssymb}
\usepackage{amsmath}
\usepackage{amsthm}
\usepackage{colortbl}
\usepackage{fancyvrb}
\usepackage{xcolor}
\usepackage{graphicx}
\usepackage{thmtools}
\usepackage{hyperref, cleveref}
\newcommand{\setn}[1]{[#1]}
\newcommand{\oddsetn}[1]{[#1]^{odd}}

\newcommand{\naturals}{\mathbb{N}}
\newcommand{\setnocond}[1]{\{#1\}}
\newcommand{\setcond}[2]{\{\, #1 \mid #2 \,\}}

\newcommand{\wordletter}[2]{#1{[#2]}}

\newcommand{\odd}[1]{\textit{odd}(#1)}
\newcommand{\even}[1]{\textit{even}(#1)}
\newcommand{\vertex}[2]{\langle {#1}, {#2} \rangle}

\newcommand{\alphabet}{\Sigma}

\newcommand{\infwords}{\alphabet^\omega}

\newcommand{\langsymb}[0]{\mathcal{L}}
\newcommand{\lang}[1]{\langsymb(#1)}
\newcommand{\inflang}[1]{\langsymb(#1)}

\newcommand{\states}{Q}
\newcommand{\trans}{\delta}
\newcommand{\ntrans}{\delta_1}
\newcommand{\jtrans}{\delta_t}
\newcommand{\dtrans}{\delta_2}
\newcommand{\inits}{I}
\newcommand{\acc}{F}
\newcommand{\run}{\rho}

\newcommand{\range}[1]{\setn{#1}}
\newcommand{\irange}[1]{{\langle #1 \rangle}}

\newcommand{\A}{\mathcal{A}}

\newcommand{\D}{\mathcal{D}}
\newcommand{\R}{\mathcal{R}}

\newcommand{\bigO}{\mathcal{O}}

\newcommand{\cover}[2]{\leq^{#1}_{#2}}

\newcommand{\buchi}{B\"uchi\xspace}
\newcommand{\rkc}{\textsf{RKC}\xspace}
\newcommand{\slc}{\textsf{SLC}\xspace}

\newcommand{\size}[1]{|#1|}

\newcommand{\infstates}[1]{\mathit{inf}({#1})}

\newtheorem{innercustomdef}{Definition}
\newenvironment{customdef}[1]
{\renewcommand\theinnercustomdef{#1}\innercustomdef}{\endinnercustomdef}

\newtheorem{definition}{Definition}
\newtheorem{lemma}{Lemma}
\newtheorem{theorem}{Theorem}
\newtheorem{corollary}{Corollary}
\newtheorem{remark}{Remark}
\newtheorem{example}{Example}

\title{On the Power of Unambiguity in \buchi Complementation}
\author{Yong Li
\institute{State Key Laboratory of Computer Science, Institute of Software, \\ Chinese Academy of Sciences}
\email{liyong@ios.ac.cn}
\and
Moshe Y. Vardi
\institute{Rice University}
\email{vardi@cs.rice.edu}
\and 
Lijun Zhang
\institute{State Key Laboratory of Computer Science, Institute of Software, \\ Chinese Academy of Sciences}
\institute{University of Chinese Academy of Sciences}
\institute{Institute of Intelligent Software, Guangzhou}
\email{zhanglj@ios.ac.cn}
}
\def\titlerunning{On the Power of Unambiguity}
\def\authorrunning{Y. Li, M. Y. Vardi \& L. Zhang}

\begin{document}
\maketitle

\begin{abstract}
In this work, we exploit the power of \emph{unambiguity} for the complementation problem of \buchi automata by utilizing reduced run directed acyclic graphs (DAGs) over infinite words, in which each vertex has at most one predecessor.
We then show how to use this type of reduced run DAGs as a \emph{unified tool} to optimize \emph{both} rank-based and slice-based complementation constructions for \buchi automata with a finite degree of ambiguity.
As a result, given a \buchi automaton with $n$ states and a finite degree of ambiguity, the number of states in the complementary \buchi automaton constructed by the classical rank-based and slice-based complementation constructions can be improved, respectively, to $2^{\bigO(n)}$ from $2^{\bigO( n \log n)}$ and to $\bigO(4^n)$ from $\bigO( (3n)^n)$.
\end{abstract}

\section{Introduction}\label{sec:intro}

The complementation of nondeterministic \buchi automata on words (NBWs)~\cite{buchi90decision} is a classic problem for NBWs and is the fundamental construction for many other important questions such as model checking\cite{DBLP:conf/lics/VardiW86} and program-termination analysis \cite{DBLP:conf/cav/HeizmannHP14}.
For instance, the complementation of NBWs is particularly valuable to model checking, when both the system $A$ and the specification $B$ are given as NBWs.
A model-checking problem essentially asks whether the behavior of the system $A$ satisfies the specification $B$.
In automata-based model checking~\cite{DBLP:conf/lics/VardiW86} framework,
this model-checking problem reduces to a language-containment problem between the NBWs $A$ and $B$.
The standard approach to solving the language-containment problem between $A$ and $B$ relies on the complementation of $B$;
one first has to construct a complementary automaton $B^c$ such that $\lang{B^c} = \infwords \setminus \lang{B}$ and then checks language emptiness of $\lang{A}\cap\lang{B^c}$. Various implementations of this approach with optimizations~\cite{DBLP:conf/cav/AbdullaCCHHMV10, DBLP:conf/concur/AbdullaCCHHMV11, DBLP:journals/corr/abs-0902-3958, DBLP:journals/lmcs/ClementeM19} have been proposed to improve its practical performance. All the implementations above, however, directly or indirectly, resort to constructing $B^c$, which can be exponentially larger than $B$~\cite{Schewe09,Yan/08/lowerComplexity}.

The complementation of \buchi automata is also a key component in the automata-based program-termination checking framework proposed in  \cite{DBLP:conf/cav/HeizmannHP14}.
This framework starts with a termination proof of a sample path of the given program and then generalizes that path to a \buchi automaton, whose language (by construction) represents a set of terminating paths.
All these terminating paths are then removed from the program.
The removal of those paths is done by automata difference operation, involved with \buchi complementation and intersection.
By iteratively removing terminating paths, the framework may obtain an empty program in the end, thus also proving the termination of the program.
It has been shown in \cite{DBLP:conf/pldi/ChenHLLTTZ18} that efficient complementation algorithms for \buchi automata can significantly improve the performance of the program-termination checking framework.

In this work, we focus on the complementation of NBWs.
The complexity for complementing NBWs has been proved to be $\Omega((0.76n)^n)$~\cite{Yan/08/lowerComplexity,Schewe09}.
A classic line of research on complementation aims at developing optimal (or close to optimal) complementation algorithms. Currently there are mainly four types of practical complementation algorithms for NBWs, namely \emph{Ramsey-based}~\cite{sistla1987complementation}, \emph{determinization-based}~\cite{safra1988complexity}, \emph{rank-based}~\cite{kupferman2001weak} and \emph{slice-based}~\cite{kahler2008complementation} algorithms.
These algorithms, however, all unavoidablely lead to a super-exponential growth in the size of $B^c$ in the worst case~\cite{Yan/08/lowerComplexity}. 

With the growing understanding of the worst-case complexity of those algorithms, searching for specialized complementation algorithms for certain subclasses of NBWs with better complexity has become an important line of research. For instance, complementing deterministic and semi-deterministic \buchi automata can be done in $\bigO(n)$~\cite{DBLP:journals/jcss/Kurshan87} and $\bigO(4^n)$~\cite{Blahoudek16}, respectively. Here we follow this line of research and aim at a subclass of NBWs with restricted nondeterminism. This type of NBWs is important, as in some contexts, especially in probabilistic 
verification, unrestricted nondeterminism in the automata representing the properties is problematic for the verification procedure. For instance, general NBWs cannot be used directly to verify properties over Markov chains, as they will cause imprecise probabilities in the product of the system and the property \cite{BustanRV04}. In turn, it is often necessary to construct their more deterministic counterparts in terms of other types of automata for the properties, for instance semi-deterministic \buchi automata or deterministic Rabin automata, which, however, adds exponential blowups of states \cite{CourcoubetisY95}.

To avoid state-space exponential blowup, earlier work sought to use of a type of automata called \emph{unambiguous nondeterministic \buchi automata} (UNBWs) in probabilistic verification~\cite{baier2016markov, DBLP:conf/setta/LiLTHZ16}, as UNBWs can be exponentially smaller than their equivalent deterministic automata~\cite{baier2016markov}. UNBWs~\cite{carton2003unambiguous} are a subclass of NBWs that accept with at most one run for each word, while their equivalent NBWs may have more than one accepting run, or even infinitely many accepting runs.
For example, by taking advantage of their unambiguity, the language-containment problem of certain proper subclasses of UNBWs has been proved to be solvable in polynomial time~\cite{bousquet2010equivalence}, while this problem is PSPACE-complete for NBWs~\cite{DBLP:conf/cav/KupfermanV96a}.

The complementation problem of a more general class than UNBWs, called \emph{finitely ambiguous nondeterministic \buchi automata} (FANBWs), which accept with finitely many runs for each word, was shown to be doable in $\bigO(5^n)$~\cite{rabinovich18}, in contrast to $2^{\Omega(n\log n)}$ for general NBWs \cite{Schewe09}.
Further, checking whether an NBW is an FANBW can be done in polynomial time~\cite{LodingP18}. Therefore, once an FANBW has been identified, the specialized complementation construction for FANBWs can be applied. In this paper, we focus here on an in-depth study of the complementation problem for FANBWs.

Our main technical tool is a construction of reduced directed acyclic graphs (DAGs) of runs of FANBWs over infinite words called \emph{co-deterministic run DAGs}, in which each vertex has at most one predecessor.
This type of co-deterministic run DAGs is previously introduced in~\cite{DBLP:journals/corr/abs-1110-6183,rabinovich18} and we defer the comparison of \cite{DBLP:journals/corr/abs-1110-6183,rabinovich18} and our construction to related works section.
We show that such co-deterministic run DAGs can be used to simplify and improve both the classical rank-based and slice-based complementation constructions. Our contributions are the following.

\begin{itemize}
\item 
First, we apply the co-deterministic run DAGs of FANBWs over infinite words, as a \emph{unified tool} to show how unambiguity works in \buchi complementation, to optimizing \emph{both} rank-based complementation (\rkc) and slice-based complementation (\slc).
\item
Second, we show that the construction of co-deterministic run DAGs in different complementation algorithms \cite{TsaiFVT14} helps to achieve simpler and theoretically better complementation algorithms for FANBWs. Given an FANBW with $n$ states, we show that the number of states of the complementary NBW constructed by the classical \rkc and \slc constructions can be improved, respectively, to $2^{\bigO(n)}$ from $2^{\bigO( n \log n)}$ and to $\bigO( 4^n)$ from $\bigO( (3n)^n)$.

\item
Finally, we reveal that \slc is basically an algorithm based on the construction of co-deterministic run DAGs and a specialized complementation algorithm for FANBWs. We also provide a language containment relation between states in the complementary NBWs of FANBWs, which can be used to improve the containment checking between an NBW and an (FA)NBW and also to reduce the number of redundant states in the complementary NBW.
\end{itemize}

\paragraph*{Related work.}
Run DAGs were introduced in \cite{kupferman2001weak} and co-deterministic run DAGs were first described in~\cite{DBLP:journals/corr/abs-1110-6183}.
In~\cite{DBLP:journals/corr/abs-1110-6183}, Fogarty and Vardi exploit co-deterministic run DAGs to complement \emph{reverse deterministic} \buchi automata with \rkc and the Ramsey-based algorithm, while we consider \rkc and \slc in this work.
In a reverse deterministic \buchi automaton, each state has only one predecessor for each letter, for which all run DAGs are already co-deterministic, as explained in Section~\ref{ssec:improved-rank-algo}, while the run DAGs of FANBWs may not be co-deterministic without our construction described in Section~\ref{sec:reduced-run-dag}.

Later co-deterministic run DAGs were constructed in~\cite{rabinovich18} under the name of \emph{narrow forest} for complementing FANBWs with the \slc construction only. Here we present it as co-deterministic run DAGs to serve as a unified tool for explaining concepts in both \rkc and \slc constructions. A subtle difference between the construction of co-deterministic run DAGs in~\cite{rabinovich18} and ours is as follows. To construct a co-deterministic run DAG over $w \in \infwords$, Rabinovich~\cite{rabinovich18} makes use of a transducer $\mathcal{T}$ that chooses one predecessor for each vertex at current level, while our construction utilizes a transition function to make the sets of successors of each pair of vertices at current level disjoint with each other, as given in Definition~\ref{def:edge-relation-e}.

More significantly, for complementation, we applied co-deterministic run DAGs to \emph{both} \rkc~\cite{kupferman2001weak} and \slc as presented in~\cite{DBLP:conf/birthday/VardiW08}.
(The complementation construction proposed in~\cite{rabinovich18} is a variant of \slc as introduced in~\cite{kahler2008complementation}.)
The comparison of the construction in~\cite{rabinovich18} and our improvement over \slc is as follows.
First, the complementary NBW constructed in~\cite{rabinovich18} is a UNBW with at most $\bigO(5^n)$ states;
this complementary NBW is the product automaton of the transducer $\mathcal{T}$, a \buchi automaton $\mathcal{C}$ for expressing unambiguity and a \buchi automaton $\D$ for accepting all possible ways to construct co-deterministic DAGs over $w \notin \lang{\A}$.
Our complementary NBW is not required to be a UNBW, since we are interested in complementation for containment checking.
Thus, the bound of $\bigO(5^n)$ in~\cite{rabinovich18} is exponentially higher than the bound of $\bigO(4^n)$ in this work.
Indeed, the product automaton of $\mathcal{T}$ and $\D$ in~\cite{rabinovich18} does yield a complementary NBW with $\bigO(4^n)$ states, but this construction and complexity were not explicitly given in~\cite{rabinovich18}.

Second, the construction in~\cite{rabinovich18} and our \slc-based construction are both based on reduced DAGs in which each vertex has at most one predecessor. These two constructions, however, are technically different and have different emphases. The construction in~\cite{rabinovich18} aims at building a complementary NBW $\A^c$ that is unambiguous, based on building product of three automata, in which each automaton fulfills part of the desired functionality for $\A^c$. For instance, $\mathcal{C}$ takes care of unambiguity and $\mathcal{D}$ obtains the complementary language. While our focus is on a complementation construction for containment checking. In contrast to building product automata in~\cite{rabinovich18}, our construction in Section~\ref{ssec:opt-slice-algo} takes a tuple of sets of states of $\A$ as a state in the complementary automaton $\A^c$ of $\A$ and performs directly on those tuples for computing successors on-the-fly, following the idea of the NCSB complementation for semi-deterministic \buchi automata in~\cite{Blahoudek16}. Various subsumption relations have been proposed in~\cite{DBLP:conf/pldi/ChenHLLTTZ18} for this representation of states in the NCSB complementation and help to reduce the number of states in $\A^c$, even improving termination analysis of programs. Inspired by~\cite{DBLP:conf/pldi/ChenHLLTTZ18}, we can also define a subsumption relation between states in $\A^c$ (see Corollary~\ref{coro:inclusion}) by our construction, which can be used to improve the containment checking between an NBW and an (FA)NBW and to reduce the number of states in $\A^c$.

\paragraph*{Organization of the paper.}
In the remainder of this paper, we first recap some definitions about \buchi automata in Section~\ref{sec:preliminaries} and then introduce the concept of co-deterministic run DAGs in Section~\ref{sec:reduced-run-dag}. We present our improved algorithms for the rank-based and slice-based algorithms in Section~\ref{sec:rank-based} and Section~\ref{sec:slice-based}, respectively. Finally we conclude the paper with some future works in Section~\ref{sec:conclusion}.

\section{Preliminaries}\label{sec:preliminaries}

We fix an \emph{alphabet} $\alphabet$.
A \emph{word} is an infinite sequence $w$ of letters in $\alphabet$.
We denote by $\infwords$ the set of all (infinite) words.
A \emph{language} is a subset of $\infwords$.
Let $L$ be a language and the complement language of $L$ is denoted by $L^c$, i.e., $L^c = \infwords \setminus L$.
Let $\run$ be a sequence of elements:
we denote by $\wordletter{\run}{i}$ the $i$-th element of $\run$.
Let $n$ be a natural number;
we denote by $\range{n}$ the set of numbers $\setnocond{0, 1, \cdots, n}$, $\oddsetn{n}$ the set of odd numbers in $\range{n}$ and $\irange{n}$ the set of numbers $\range{n} \setminus \setnocond{0}$.

A \emph{nondeterministic B\"uchi automaton on words} (NBW) is a tuple $\A = (\states, \inits, \trans, \acc)$, where $\states$ is a finite set of states, $\inits \subseteq \states$ is a set of initial states, $\trans : \states \times \alphabet \rightarrow 2^{\states}$ is a transition function and $\acc \subseteq \states$ is a set of accepting states. We extend $\trans$ to sets of states, by letting $\trans(S, a) = \bigcup_{q \in S} \trans(q, a)$. We assume that each NBW $\A$ is \emph{complete} in the sense that for each state $q \in \states$ and $a \in \alphabet$, $\trans(q, a) \neq \emptyset$. 
A \emph{run} of $\A$ on a word $w$ is an infinite sequence of states $\run = q_{0} q_{1}\cdots$ such that $q_{0} \in \inits$ and for every $i > 0$, $q_{i} \in \trans(q_{i-1}, a_{i})$. We denote by $\infstates{\run}$ the set of states that occur infinitely often in the run $\run$. A word $w \in \infwords$ is \emph{accepted} by $\A$ if there exists a run $\run$ of $\A$ over $w$ such that $\infstates{\run} \cap \acc \neq \emptyset$.
We denote by $\inflang{\A}$ the \emph{language} recognized by $\A$, i.e., the set of words accepted by $\A$.

Let $\A$ be an NBW.
A complementary NBW of $\A$ is an NBW that accepts the complementary language $\infwords \setminus \inflang{\A}$ of $\inflang{\A}$;
we denote by ${\A}^S$ the automaton $(\states, S, \trans, \acc)$ obtained from $\A$ by setting its initial state set to the set $S \subseteq \states$.
In particular, we use $\A^q$ as the shorthand for $\A^{\setnocond{q}}$.
We say a state $q$ of $\A$ \emph{subsumes} a state $q'$ of $\A$ if $\lang{A^{q'}} \subseteq \lang{A^{q}}$.
We classify $\A$ into following types of NBWs according to their transition structures:
(1) \emph{nondeterministic} if $\size{\inits} > 1$ or $\size{\trans(q, a)} > 1$ for a state $q \in \states$ and $a \in \alphabet$,
(2) \emph{deterministic} if $\size{\inits} = 1$ and for each $q \in \states$ and $a \in \alphabet$, $\size{\trans(q, a)} \leq 1$, and
(3) \emph{reverse deterministic} if for each state $q' \in \states$, $\A$ has at most one state $q$ for each $a \in \alphabet$ such that $q ' = \trans(q, a)$.

From the perspective of the number of accepting runs of $\A$, we have following types of NBWs.
\begin{definition}
\label{def:fanbw}
Let $\A$ be an NBW and $k$ a positive integer.
We say $\A$ is
(1) finitely ambiguous (an FANBW) if for each $w \in \lang{\A}$, the number of accepting runs of $\A$ over $w$ is finite; and
(2) $k$-ambiguous if for each $w \in \lang{\A}$, the number of accepting runs of $\A$ over $w$ is no greater than $k$, and (3) unambiguous if it is $1$-ambiguous.
\end{definition}

By Definition~\ref{def:fanbw}, it holds that both $k$-ambiguous NBWs and unambiguous NBWs are special classes of FANBWs. For instance, the NBW $\A$ depicted in Figure~\ref{fig:example} is a 2-ambiguous NBW, thus also an FANBW, as $(q_{0})^{i+1} q_{1}^{\omega}$ and $(q_{0})^{i+1} q_{2} q_{1}^{\omega}$ are the only two accepting runs for accepting word $a^{i} b^{\omega} \in \lang{\A}$ where $i \geq 0$.

\begin{figure}
\centering
\includegraphics[scale=0.9]{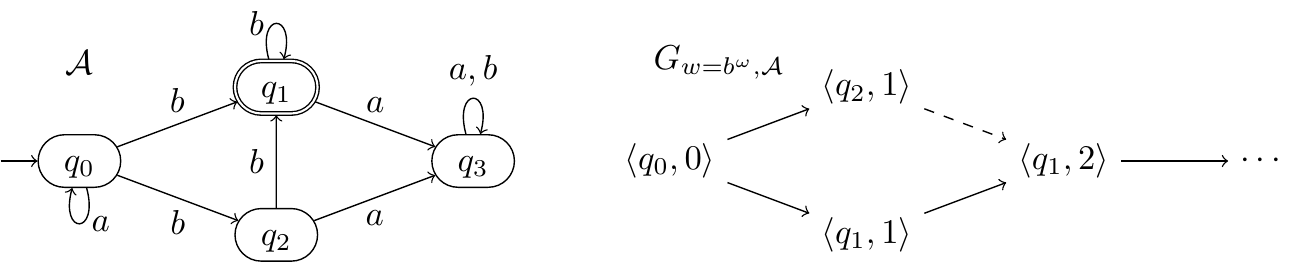}
\caption{An FANBW $\A$ with $\inits = \setnocond{q_{0}}$ and $\acc = \setnocond{q_{1}}$ and the run DAG $G_{w, \A}$ over $b^{\omega}$. }\label{fig:example}
\end{figure}

\section{Co-Deterministic Run DAGs for FANBWs}
\label{sec:reduced-run-dag}
In this section, we first describe the concept of run DAGs of an NBW over a word $w$, introduced in \cite{kupferman2001weak}.
We then describe the co-deterministic run DAGs for FANBWs as a unified tool for both \rkc and \slc constructions by making use of the finite ambiguity in FANBWs.
In the remainder of the paper, we use DAGs as the shorthand for run DAGs.

Let $\A = (\states, \inits, \trans, \acc)$ be an NBW and $w = a_0 a_1 \cdots$ be an infinite word. The DAG $G_{w, \A} = \langle V, E\rangle$ of $\A$ over $w$ is defined as follows:
\begin{itemize}
\item Vertices: $V \subseteq \states \times \naturals$ is the set of vertices  $\bigcup_{l \geq 0}V_{l} \times \setnocond{l}$ where $V_{0} = \inits$ and $V_{l + 1} := \trans(V_{l}, a_l)$ for every $l \geq 0$.
\item Edges: There is an edge from $\langle q, l\rangle$ to $\langle q', l'\rangle$ iff $l' = l + 1$ and $q' \in \trans(q, a_l)$. 
\end{itemize}

A vertex $\vertex{q}{l}$ is said to be on level $l$ and there are at most $\size{\states}$ states on each level. A vertex $\vertex{q}{l}$ is an \emph{$\acc$-vertex} if $q \in \acc$.
A finite/infinite sequence of vertices $\gamma = \vertex{q_{0}}{0}\vertex{q_{1}}{1} \cdots$ is called a \emph{branch} of $G_{w, \A}$ if $q_{0} \in \inits$ and for each $l \geq 0$, there is an edge from $\vertex{ q_{l}}{ l}$ to $\vertex{ q_{l + 1}}{ l+ 1}$.
An \emph{$\omega$-branch} of $G_{w, \A}$ is a branch of infinite length.
A \emph{fragment} $\vertex{ q_{l}}{ l}\vertex{ q_{l+1}}{ l+1} \cdots$ of $\gamma$ is said to be a branch from the vertex $\vertex{ q_{l}}{ l}$; a fragment $\vertex{ q_{l}}{ l} \cdots \vertex{ q_{l+k}}{ l+k}$ of $\gamma$ is said to be a \emph{path} from $\vertex{ q_{l}}{ l}$ to $\vertex{ q_{l+k}}{ l+k}$, where $k \geq 1$. A vertex $\vertex{ q_{j}}{ j}$ is \emph{reachable} from $\vertex{ q_{l}}{ l}$ if there is a path from $\vertex{ q_{l}}{ l}$ to $\vertex{ q_{j}}{ j}$. We call a vertex $\vertex{q}{l}$ \emph{finite} if there are no $\omega$-branches in $G_{w, \A}$ starting from $\vertex{q}{l}$; and we call a vertex $\vertex{q}{l}$  \emph{$\acc$-free} if it is not finite and no $\acc$-vertices are reachable from $\vertex{q}{l}$ in $G_{w, \A}$.

There is a bijection between the set of runs of $\A$ on $w$ and the set of $\omega$-branches in $G_{w,\A}$. To a run $\run = q_{0} q_{1} \cdots$ of $\A$ over $w$ corresponds an $\omega$-branch $\hat{\run} = \langle q_{0}, 0 \rangle\langle q_{1}, 1 \rangle \cdots$. Therefore, $w$ is accepted by $\A$ if and only if there exists an $\omega$-branch in $G_{w, \A}$ that visits $\acc$-vertices infinitely often; we say that such an $\omega$-branch is \emph{accepting}; $G_{w, \A}$ is accepting if and only if there exists an accepting $\omega$-branch in $G_{w, \A}$.

Assume that $\A$ is an FANBW. Then an accepting $\omega$-branch in $G_{w, \A}$, if exists, only merges with other (accepting) $\omega$-branches for finitely many times.
That is, there exists a level $k \geq 1$ such that all vertices after level $k$ on an accepting $\omega$-branch have exactly one predecessor;
we call the level $k$ a \emph{separate level}.
We formalize this property of $G_{w, \A}$ in Lemma~\ref{lem:vertices-one-pred}.

\begin{restatable}[Separate Levels of Accepting DAGs of FANBWs]{lemma}{separateLevelDAG}
\label{lem:vertices-one-pred}
Let $\A$ be an FANBW and $G_{w, \A}$ the accepting DAG of $\A$ over $w \in \lang{\A}$.
Then there must exist a separate level $k \geq 1$ in $G_{w,\A}$.
\end{restatable}
\begin{proof}
Since $\A$ is an FANBW, there are only finitely many accepting $\omega$-branches in $G_{w, \A}$.
Therefore, an accepting $\omega$-branch in $G_{w, \A}$ only merges with other (accepting) $\omega$-branches for finitely many times. It follows that given an accepting $\omega$-branch $\hat{\run}$ in $G_{w, \A}$, there must exist a separate level $h \geq 1$ such that each vertex $\wordletter{\hat{\run}}{i}$ with $i \geq h$ has exactly one predecessor. Otherwise, there will be infinitely many accepting branches, contradicting with the assumption that $\A$ is an FANBW. Assume that there are $m < \infty$ accepting $\omega$-branches in $G_{w, \A}$. Then we can set the separate level $k$ of $G_{w,\A}$ to $\max\setcond{h_{i}}{ 1 \leq i \leq m}$ where $h_{i}$ is the separate level index of $i$-th accepting $\omega$-branch. 
\end{proof}

For instance, the separate level is $2$ in the accepting DAG $G_{w, \A}$ of $\A$ over $b^{\omega}$ in Figure~\ref{fig:example}, as each vertex $\vertex{q_{1}}{i}$ with $i \geq 3$ only has the predecessor $\vertex{q_{1}}{i - 1}$.

It follows immediately from  Lemma~\ref{lem:vertices-one-pred} that for each vertex $v$ in $G_{w, \A}$ with more than one incoming edge, keeping only one of incoming edges of $v$ will not change whether $G_{w, \A}$ is accepting.
Assume that $\states = \setnocond{s_{1}, s_{2}, \cdots , s_{n}}$.
We define an \emph{edge-reduced} DAG $G^{e}_{w, \A} = \langle V, E^{e}\rangle$ called co-deterministic DAG,  in which each vertex only has at most one predecessor with the following policy for removing edges:
if there is a vertex with multiple incoming edges in $G_{w,\A}$, we only keep the incoming edge from the predecessor with the minimal index.
Formally, the definition of edges in $G^{e}_{w, \A}$ is given as follows. 
\begin{itemize}
\item Edges.
There is an edge from $\langle s_{k}, l \rangle$ to $\langle s', l'\rangle$ iff $l' = l + 1$ and $k = \min \setcond{p \in \irange{n}}{ s' \in \trans(s_{p}, a_{l + 1})}$. 
\end{itemize}
Lemma~\ref{lem:fanbw-run-acc} ensures that $G^{e}_{w,\A}$ is accepting if $G_{w,\A}$ is accepting.
\begin{restatable}[Acceptance of Co-deterministic DAGs]{lemma}{fanbwrunacc}
\label{lem:fanbw-run-acc}
Assume that $\A$ is an FANBW.
Let $G^{e}_{w, \A}$ be the co-deterministic DAG of $\A$ over a word $w \in \infwords$.
Then $w$ is accepted by $\A$ if and only if $G^{e}_{w, \A}$ is accepting.
\end{restatable}
\begin{proof}
The proof is trivial when $G_{w, \A}$ is nonaccepting.
Assume that $G_{w, \A}$ is accepting.
Let $\hat{\run}$ be an accepting $\omega$-branch and $k$ the separate level defined in Lemma~\ref{lem:vertices-one-pred}.
According to Lemma~\ref{lem:vertices-one-pred}, the $\omega$-branch from $\wordletter{\hat{\run}}{k+1}$ must be accepting.
Moreover, $\wordletter{\hat{\run}}{k+1}$ is reachable from an initial vertex $\vertex{q}{0}$ with $q \in \inits$.
Then there must exist an accepting $\omega$-branch in $G^{e}_{w, \A}$ if $G_{w, \A}$ is accepting.
Thus we conclude that $w $ is accepted by $\A$ if and only if $G^{e}_{w, \A}$ is accepting.
\end{proof}

For instance, the co-deterministic DAG of $G_{w,\A}$ in Figure~\ref{fig:example} is still accepting after deleting the edge from $\vertex{q_{2}}{1}$ to $\vertex{q_{1}}{2}$, as denoted by the dashed arrow.

By removing redundant edges, we can now define a reduced transition function $\trans^e_{w,\ell}: 2^\states \times \Sigma \rightarrow 2^{\states}$ over the levels in $G^{e}_{w, \A}$.
\begin{definition}[Transition Function for Co-deterministic DAGs]
\label{def:edge-relation-e}
Given the set of states $S \subseteq \states$ at level $\ell$ of $G^{e}_{w, \A}$ and let $S' = \trans(S, \wordletter{w}{\ell})$ be the set of states at level $\ell + 1$.
Define $S_{min} = \setcond{q_{m} \in S}{ m \in \min \setcond{k \in \irange{n} }{ q' \in \trans(q_{k}, \wordletter{w}{\ell})}, q' \in S'}$ as the minimal set of predecessors of $S'$. Then, for a set of states $S_1 \subseteq S$, we define $\trans^{e}_{w,\ell}(S_{1}, \wordletter{w}{\ell}) = \trans(S_{1} \cap S_{min} , \wordletter{w}{\ell})$. We call $\trans^{e}_{w,\ell}$ the \emph{reduced transition function at level}~$\ell$ in $G^{e}_{w, \A}$. 
\end{definition}

\begin{example}
\label{example:trans-not-on-level}
Consider again $G_{b^{\omega}, \A}$ in Figure~\ref{fig:example} and let $S = \setnocond{q_{1}, q_{2}}$ at level~$1$: we have $S' = \trans(S, b) = \setnocond{q_{1}}$ and $S_{min} = \setnocond{q_{1}}$. Let $\trans^{e}_{b^{\omega}, 1}$ be the reduced transition function at level $1$ defined from $\trans$ in Definition \ref{def:edge-relation-e}. It follows that $\trans^{e}_{b^{\omega}, 1}(\setnocond{q_{1}}, b) = \trans(\setnocond{q_{1}}\cap S_{min}, b) = \setnocond{q_{1}}$ and $\trans^{e}_{b^{\omega}, 1}(\setnocond{q_{2}}, b) = \trans(\setnocond{q_{2}}\cap S_{min}, b) = \emptyset$.
\end{example}

In general, the reduced transition function $\trans^e_{w,\ell}$ may seem to depend on the level $\ell$ and the word $w$ yielding the edge connections between vertices at levels $\ell$ and $\ell+1$ in $G^e_{w, \A}$. We claim, however,  that in Definition \ref{def:edge-relation-e}, $\trans^{e}_{w,\ell}$ is not dependent on the level number $\ell$ and the word $w$, due to our specific choice of the set $S_{min}$. Thus, we can omit the level $\ell$ and $w$ in our notion $\trans^{e}$.

\begin{lemma}
\label{lem:independent-reduced-trans}
Let $S \subseteq \states$ and $b$ be the set of states and the input letter at the level $\ell_{1}$ in $G^e_{w_{1}, \A}$ and at the level $\ell_{2}$ in $G^e_{w_{2}, \A}$, respectively.
Then $\trans^e_{w_{1},\ell_{1}}$ of $G^e_{w_{1},\A}$ and $\trans^e_{w_{2},\ell_{2}}$ of $G^e_{w_{2},\A}$ are identical regardless of their different level numbers and infinite words.
\end{lemma}
\begin{proof}
According to Definition \ref{def:edge-relation-e}, we can let $\wordletter{w_{1}}{\ell_{1}} = \wordletter{w_{2}}{\ell_{2}} = b$. Then all the subsequent computations defined for both $\trans^e_{w_{1}, \ell_{1}}$ and $\trans^e_{w_{2}, \ell_{2}}$ only depend on the set of states $S$ and the input letter $b$, not their level numbers and the entire infinite words.
Thus we complete the proof.
\end{proof}

Because of Lemma \ref{lem:independent-reduced-trans}, we can just use the reduced transition function $\trans^e$ with respect to the set of states $S$ and the input letter $b$ at a level in the construction of complementary NBWs of FANBWs (see Definitions \ref{def:complement-rank-kv-fanbws} and \ref{def:unbw-complement}). We remark that one can define different co-deterministic DAGs from those constructed in this work.
This is illustrated in the following example.
\begin{example}[$\trans^e_{w,\ell}$ depending on $\ell$]
\label{example:trans-on-level}
Consider $G_{a^{\omega}, \A}$ in Figure~\ref{fig:num-of-omega-branches} and let $S_{\ell} = \setnocond{q_{0}, q_{1}, q_{2}}$ at level~$\ell \geq 2$:
we have $S_{\ell} = \trans(S_{\ell}, a)$ as the set of states on each level $\ell \geq 2$.
Rather than keeping the predecessor with the minimal index of a state in $S_{min}$ (see Definition \ref{def:edge-relation-e}), one can define $S_{\ell, min}$ as $S_{min}$ depending on the level $\ell$ as follows.
We define $S_{\ell, min} = \setnocond{q_{0}, q_{1}}$ when $\ell$ is an odd number and $S_{\ell, min} = \setnocond{q_{0}, q_{2}}$ otherwise.
That is, we keep the predecessor $q_{1}$ of $q_{2}$ at odd levels and $q_{2}$ at even levels.
Let $\trans^{e}_{a^{\omega}, \ell}$ be the reduced transition function at level $\ell$. It follows that $\trans^{e}_{a^{\omega}, \ell}(\setnocond{q_{1}}, a) = \trans(\setnocond{q_{1}}\cap S_{\ell, min}, a) = \setnocond{q_{2}}$ when $\ell$ is odd and $\trans^{e}_{a^{\omega}, \ell}(\setnocond{q_{1}}, a) = \emptyset$ otherwise.
Clearly, the definition of $\trans^e_{a^{\omega}, \ell}$ is dependent on the level $\ell$ and the resulting co-deterministic DAG is different from the one depicted in Figure~\ref{fig:num-of-omega-branches} where dashed arrows denote the removed edges.  
\end{example}

In the remainder of the paper, we may write $\trans^{e}(q, b)$ instead of $\trans^{e}(\setnocond{q}, b)$ for an input singleton set $\setnocond{q}$.
The transition function $\trans^{e}$ will be used in the complementation of FANBWs since the complementation essentially constructs DAGs and then identifies accepting DAGs.

One can verify that each vertex in the co-deterministic DAG $G^{e}_{w, \A}$ of $\A$ over $w$ has at most one predecessor. It follows that the number of $\omega$-branches in a non-accepting/accepting $G^{e}_{w,\A}$ is at most $\size{\states}$, as stated in Lemma~\ref{lem:unbw-rdag-num-stable}.
\begin{restatable}[Finite Number of $\omega$-Branches in Co-deterministic DAGs]{lemma}{finiteNumOmegaBranches}
\label{lem:unbw-rdag-num-stable}
Let $G^{e}_{w, \A}$ be a co-deterministic DAG of $\A$ over $w$.
Then the number of $\omega$-branches in $G^{e}_{w, \A}$ is at most $\size{\states}$.
\end{restatable}
\begin{proof}
Let $m_{i}$ with $i \geq 0$ be the number of vertices which are in the $\omega$-branches (not in all branches) on level $i$.
For instance, $m_{i} = 1$ for each $i \geq 1$ in Fig.~\ref{fig:example} while the number of vertices on level $1$ is $2$.
Since each vertex in $G^{e}_{w, \A}$ has only one predecessor, we have that $m_{0} \leq m_{1} \leq m_{2} \leq \cdots$, i.e., the number of vertices in $\omega$-branches on each level does not decrease over the levels.
In addition, there are at most $\size{\states}$ states on each level.
Thus there are at most $\size{\states}$ $\omega$-branches since we have $m_{i} \leq \size{\states}$ for each $i \geq 0$.
\end{proof}

\begin{figure}
\begin{center}
\includegraphics[scale=1]{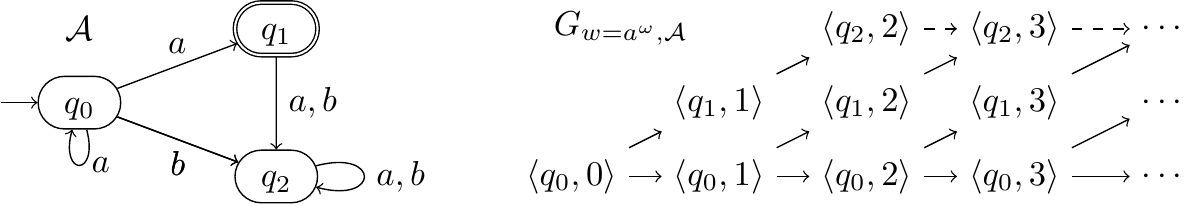}
\end{center}
\caption{Another FANBW $\A$ with $\inits = \setnocond{q_{0}}$ and $\acc = \setnocond{q_{1}}$ and the run DAG $G_{w, \A}$ over $a^{\omega}$}. 
\label{fig:num-of-omega-branches}
\vspace{-0.25cm}
\end{figure}

Consider the DAG $G_{w,\A}$ in Figure~\ref{fig:num-of-omega-branches}:
one can verify that there are infinitely many $\omega$-branches in the non-reduced DAG $G_{w,\A}$ over $a^{\omega}$;
while for the co-deterministic DAG of $G_{w,\A}$ where removed edges are marked with dashed arrows, there is only one $\omega$-branch $\vertex{q_{0}}{0}\vertex{q_{0}}{1} \cdots \vertex{q_{0}}{l} \cdots$.

After redundant edges have been cut off, we obtain a DAG $G^{e}_{w,\A}$ with a finite number of $\omega$-branches.
Thus if $w \notin \lang{\A}$, there must exist a maximum level $l > 0$ among those $\omega$-branches such that each $\acc$-vertex $\vertex{q}{l'}$ with $l' \geq l$ is finite, which can be used for identifying whether $G^{e}_{w,\A}$ is accepting in the complementation of FANBWs.
We call a level $l > 0$ a \emph{stable level} in $G^e_{w,\A}$ if each $\acc$-vertex $\vertex{q}{l'}$ with $l' \geq l$ in $G^e_{w,\A}$ is finite.

\begin{restatable}[Stable Level in Nonaccepting Co-deterministic DAGs]{lemma}{stableLevelCoDetDAG}
\label{lem:unbw-rdag-nonacc-branch-stable}
Assume that $\A$ is an FANBW and $w \in \infwords$.
Let $G^{e}_{w, \A}$ be the co-deterministic DAG of $\A$ over $w$.
Then $w \notin \lang{\A}$ if and only if there exists a stable level $k > 0$ in $G^{e}_{w,\A}$.
\end{restatable}
\begin{proof}
($\Leftarrow$)
By Lemma \ref{lem:fanbw-run-acc}, if $w \in \lang{\A}$ and $\A$ is an FANBW, there exists an accepting $\omega$-branch in $G^{e}_{w, \A}$.
It follows that if $w \in \lang{\A}$, there does not exist a stable level $k$ in $G^{e}_{w, \A}$ such that each $\acc$-vertex after $k$ is finite.
Consequently, if there exists a stable level $k$ in $G^{e}_{w, \A}$, it holds that $w \notin \lang{\A}$.

($\Rightarrow$)
By Lemma~\ref{lem:unbw-rdag-num-stable}, let $m \leq \size{\states}$ be the number of $\omega$-branches in $G^{e}_{w, \A}$.
Since $w \notin \lang{\A}$, all the $\omega$-branches in $G^{e}_{w,\A}$ is nonaccepting.
Therefore, for the $i$-th $\omega$-branch $\hat{\run}_{i}$, there is a vertex $\vertex{q}{k_{i}}$ such that every vertex of $\hat{\run}_{i}$ reachable from $\vertex{q}{k_{i}}$ is not an $\acc$-vertex.  
It follows that we can set $k = \max\setcond{k_{i}}{ i \in \irange{m}}$ and thus all the $\acc$-vertices on a level after $l \geq k$ are finite and not on $\omega$-branches.
\end{proof}

Consider again the DAG $G_{w,\A}$ in Figure~\ref{fig:num-of-omega-branches}:
there does not exist a stable level in the non-reduced DAG $G_{w,\A}$ since each $\acc$-vertex $\vertex{q_{1}}{l}$ with $l \geq 1$ is not finite;
while in the co-deterministic DAG of $\A$ over $a^{\omega}$, one can verify that the stable level $k$ is $1$.

\section{Rank-Based Complementation}
\label{sec:rank-based}

We first introduce in Subsection~\ref{ssec:rank-algo} the rank-based complementation  (\rkc) proposed in \cite{kupferman2001weak}, which constructs a complementary NBW $\A^c$ for $\A$ with $2^{\bigO(n \log n)}$ states.
Then in Subsection~\ref{ssec:improved-rank-algo}, we show that if $\A$ is an FANBW, \rkc based on the construction of co-deterministic DAGs produces a complementary NBW $\A^c$ with $2^{\bigO(n)}$ states.

\subsection{Rank-Based Algorithm for NBWs}
\label{ssec:rank-algo}

\rkc was introduced by Kupferman and Vardi in \cite{kupferman2001weak} to construct a complementary NBW $\A^c$ of $\A$ by identifying the DAGs of $\A$ over nonaccepting words $ w\notin\inflang{\A}$. Intuitively, given a word $w \notin \inflang{\A}$,  all $\omega$-branches of the DAG of $\A$ over $w$ will eventually stop visiting $\acc$-vertices. Based on this observation, in order to identify the nonaccepting DAG of $\A$ over $w$, they introduced the notion of \emph{level rankings} of $G_{w, \A}$. By assigning only even ranks to $\acc$-vertices, they showed that there exists a unique ranking function that assigns ranks in $\range{2n}$ to the vertices of $G_{w, \A}$ such that $w \notin \lang{\A}$ iff all $\omega$-branches of $G_{w, \A}$ eventually get trapped in odd ranks.

We now define level rankings of a nonaccepting DAG. The level ranking of $G_{w, \A}= (V, E)$ defines a ranking function $f : V \rightarrow \range{2n}$ that satisfies the following conditions:
\begin{itemize}
\item[(i)] for each vertex $\vertex{q}{i} \in V $ if $f(\vertex{q}{i}) \in \oddsetn{2n}$, then $q \notin \acc$,
\item[(ii)] for each edge $(\vertex{q}{i}, \vertex{q'}{i+1}) \in E$, 
$f(\vertex{q'}{i+1}) \leq f(\vertex{q}{i}) $
\end{itemize}
The ranks along a branch decrease monotonically and $\acc$-vertices get only even ranks.

We now define a specific ranking function $f$ of $G_{w, \A}$ for a given word $w \notin \inflang{\A}$. We define a sequence of DAGs $G^{0}_{w, \A} \supseteq G^{1}_{w, \A} \supseteq \cdots$, where $G^{0}_{w, \A} = G_{w, \A}$, as follows.
For each $i \geq 0$, 
\begin{itemize}
\item $G^{2i + 1}_{w, \A}$ is the DAG constructed from $G^{2i}_{w, \A}$ by removing all finite vertices in $G^{2i}_{w, \A}$ and the edges associated with them, and
\item if $G^{2i + 1}_{w, \A}$ has at least one $\acc$-free vertex, then $ G^{2i+2}_{w, \A}$ is the DAG constructed from $G^{2i+1}_{w, \A}$ by removing all the $\acc$-free vertices in $G^{2i+1}_{w, \A}$ and the edges associated with them.
\end{itemize}

Recall that $\acc$-free vertices cannot reach $\acc$-vertices.
It was shown in~\cite{kupferman2001weak} that $G^{2n+1}_{w,\A}$ is empty and each vertex $\vertex{q}{l}$ is either finite in $G^{2i}_{w,\A}$ or $\acc$-free in $G^{2i+1}_{w,\A}$.
Thus the sequence of DAGs generated from the definition above defines a unique ranking function $f$ over the set of vertices in $G_{w, \A}$ inductively as follows.
For every $i \geq 0$,
\begin{itemize}
\item[(1)] $f(\vertex{q}{l}) = 2i$ for each vertex $\vertex{q}{l}$ that is finite in $G^{2i}_{w, \A}$, if exists.
\item[(2)] $f(\vertex{q}{l}) = 2i + 1$ for each $\acc$-free vertex $\vertex{q}{l}$ in $G^{2i+1}_{w, \A}$, if exists.
\end{itemize}
Consequently, we have Lemma~\ref{lem:nonacc-run-odd-ranks} for identifying nonaccepting DAGs according to~\cite{kupferman2001weak}.

\begin{lemma}[Nonaccepting DAGs~\cite{kupferman2001weak}]
\label{lem:nonacc-run-odd-ranks}
$\A$ rejects a word $w$ if and only if the unique ranking function $f$ defined in (1) and (2) above has $2n$ as maximum rank, and all $\omega$-branches of $G_{w,\A}$ eventually get trapped in odd ranks.
\end{lemma}

We have constructed a unique ranking function above for identifying nonaccepting DAGs.
To construct the complementary NBW $\A^c$ with such a ranking function, we have to guess the ranking level by level.
Since the maximum rank is $2n$, along an input word $w$, we can encode a ranking function for $G_{w,\A}$ by utilizing a \emph{level-ranking} function $f: \states \rightarrow \range{2n} \cup \setnocond{\bot}$ for the states $S$ at a level in the DAG $G_{w, \A}$ such that if $q \in S \cap \acc$, then $f(q)$ is even, and $f(q) = \bot$ if $q \in \states \setminus S$. 
\begin{definition}[Coverage Relation for Level Rankings]
\label{def:coverage-level-rnk}
Let $a $ be a letter in $\alphabet$ and $f, f'$ be two level ranking functions.
We say $f$ \emph{covers} $f'$ \emph{under} letter $a$, denoted by $f' \cover{\trans}{a} f$, when for all $q, q' \in \states$, if $f(q) \geq 0$ and $q' \in \trans(q, a)$, then $0 \leq f'(q') \leq f(q)$, otherwise $f'(q') = \bot$.
\end{definition}
Note here that $\cover{\trans}{a}$ is defined based on the transition function $\trans$.
The coverage relation indicates that the level rankings $f$ and $f'$ of two consecutive levels of $G_{w, \A}$ do not increase in ranks. We denote by $\R$ the set of all possible level ranking functions.

In order to verify that the guess about the ranking of $G_{w, \A}$ is correct, \rkc uses the \emph{breakpoint construction} proposed in \cite{miyano1984alternating}. This construction employs a set of states $O \subseteq \states$ to check that the vertices assigned with even ranks are finite.
Similarly to Lemma~\ref{lem:unbw-rdag-nonacc-branch-stable}, the nonaccepting DAG $G_{w, \A}$ with the ranking function defined in (1) and (2) eventually reaches a stable level, after which all $\acc$-vertices are finite. Hence, a breakpoint construction suffices to verify such guesses.

The formal definition of the complementary NBW $\A^c$ of the input NBW $\A$ is given below.
\begin{definition}[\hspace*{-0.1cm}\cite{kupferman2001weak}]
\label{def:complement-rank-kv}
Let $\A = (\states, \inits, \trans, \acc)$ be an NBW.
We then define an NBW $\A^c = (\states^c, \inits^c, \trans^c, \acc^c)$ of $\A$ as follows.
\begin{itemize}
\item $\states^c \subseteq \R \times 2^\states$,
\item $\inits^c = (f, \emptyset)$ where $f(q) = 2n$ if $q \in \inits$ and $f(q) = \bot$ otherwise.
\item $\trans^c$ is defined as follows:
\begin{enumerate}

\item if $O \neq \emptyset$, then $\trans^c((f, O), a) = \setcond{ (f', \trans(O, a) \setminus \odd{f'})}{ f' \cover{\trans}{a} f } $ (intuition: breakpoint $O$ only tracks vertices assigned with even ranks), 
\item if $O = \emptyset$, then $\trans^c((f, O), a) =  \setcond{ (f', \even{f'} ) }{ f' \cover{\trans}{a} f } $ (intuition: $O = \emptyset$ means all previous $\acc$-vertices with even ranks are finite, then verify new vertices with even ranks).
\end{enumerate}

\item $\acc^c= \setcond{(f, O) \in \states^c}{O = \emptyset} $.
\end{itemize}
where  $\odd{f} = \setcond{ q\in\states }{ f(q) \text{ is odd }}$ and $\even{f} = \setcond{ q \in\states}{ f(q) \text{ is even }}$.
\end{definition}
Let $w$ be a word. Intuitively, every state $(f, O)$ in $\A^c$ corresponds to a level of the DAG $G_{w, \A}$ over $w$. If $w$ is accepted by $\A^c$, i.e., $O$ becomes empty for infinitely many times, then we conclude that all the $\omega$-branches of $G_{w, \A}$ eventually get trapped in odd ranks. It follows that no branches are accepting in $G_{w, \A}$, i.e., $w \notin \inflang{\A}$. The other direction is also easy to prove and omitted here. Thus we conclude that $\inflang{\A^c} = \infwords \setminus \inflang{\A}$. Since $f \in \R$ is a function from $\states$ to $\range{2n} \cup \setnocond{\bot}$, the number of possible $f$ functions is $(2n+2)^n \in 2^{\bigO(n \log n)}$. Therefore, the number of states in $\A$ is in $2^n \times 2^{\bigO(n \log n)} \in 2^{\bigO(n \log n)}$.

\begin{lemma}[The Language and Size of $\A^c$~\cite{kupferman2001weak}]
\label{lem:correctness-rank-kv}
Let $\A$ be an NBW with $n$ states and $\A^c$ the NBW defined in Definition~\ref{def:complement-rank-kv}. Then $\inflang{\A^c} = \infwords \setminus \inflang{\A}$ and $\A^c$ has $2^{\bigO(n \log n)}$ states.
\end{lemma}

\paragraph*{Relation to Construction of Co-deterministic DAGs.} 
Assume that we have two level-rankings $f' \cover{\trans}{a} f$.
A state $q'$ in the second level can have multiple $a$-predecessors defined in the domain of $f$. Then $f'(q') \leq \min \setcond{f(q)}{f(q) \neq \bot, q' \in \trans(q, a)}$. Thus we can define a co-deterministic DAG out of $G_{w,\A}$ where each vertex only keeps one predecessor with the minimal rank in the reduced DAG, in contrast to the predecessor with minimal index in Section~\ref{sec:reduced-run-dag}. There may, however, be multiple predecessors with the minimal rank.
Consequently, the non-reduced DAG $G_{w, \A}$ can be mapped to multiple co-deterministic DAGs depending on which ranking function is defined on $G_{w, \A}$ and how predecessors are chosen.
Note here that not every resulting co-deterministic DAG of $G_{w, \A}$ described above will be accepting if $G_{w,\A}$ is accepting, since each time the edges in accepting $\omega$-branches may be deleted. Thus these co-deterministic DAGs cannot be directly applied in \rkc for general NBWs.

\subsection{Rank-Based Algorithm for FANBWs}
\label{ssec:improved-rank-algo}

In the following, we show in Lemma~\ref{lem:unbw-rdag-nonacc-branch-max-rank} that if $\A$ is an FANBW, the maximum rank of the vertices in a co-deterministic DAG of $\A$ is at most $2$. It follows that the range of $f \in \R$ is $\setnocond{0, 1, 2} \cup \setnocond{\bot}$. We thus only need the maximum rank to be $2$ rather than $2n$ for the
co-deterministic DAG $G^e_{w, \A}$ of $\A$. Therefore, the number of states in $\A^c$ is in $2^n \times 4^n \in 2^{\bigO(n)}$ when the maximum rank is $2$.

\begin{lemma}[Maximum Rank of Co-deterministic DAGs]
\label{lem:unbw-rdag-nonacc-branch-max-rank}
Assume that $\A$ is an FANBW and let $w$ be a word.
Let $G^e_{w, \A}$ be the co-deterministic DAG of $\A$ over $w$.
Then $w \notin \lang{\A}$ iff $(G^{e}_{w, \A})^3$ is empty. 
\end{lemma}
\begin{proof}
Assume that $w \notin \lang{\A}$.
Our goal is to prove that starting from $(G^{e}_{w, \A})^0 = G^{e}_{w, \A}$, $(G^{e}_{w, \A})^3$ is empty.
By Lemma~\ref{lem:unbw-rdag-nonacc-branch-stable}, there exists a stable level, say $k > 1$, such that on each level $l \geq k$, the $\acc$-vertices are finite.
Therefore, $(G^{e}_{w, \A})^1$ contains only non-$\acc$-vertices after level $k$.
It follows that $(G^{e}_{w, \A})^2$ removes all the vertices after level $k$.
Thus if $(G^{e}_{w, \A})^2$ is not empty, $(G^{e}_{w, \A})^2$ contains only finite vertices.
We then conclude that $(G^{e}_{w, \A})^3$ is empty.
The other direction is trivial.
\end{proof}

In order to set the maximum rank to $2$ in Definition~\ref{def:complement-rank-kv}, the underlying DAG $G_{w, \A}$ constructed for complementing FANBWs has to be co-deterministic. Since \rkc generates rankings level by level, we have to utilize the reduced transition function $\trans^e$ for computing successors at next level. 
For FANBWs,  the complementation construction  in Definition~\ref{def:complement-rank-kv} can be improved accordingly:
\begin{customdef}{4'}
\label{def:complement-rank-kv-fanbws}
Let $\A = (\states, \inits, \trans, \acc)$ be an FANBW.
We then define an NBW $\A^c = (\states^c, \inits^c, \trans^c, \acc^c)$, where $\states^c$ and $\acc^c$ are as in Definition~\ref{def:complement-rank-kv}, and 
$\inits^c$ and $\trans^c$ are defined by:
\begin{itemize}
    \item $\inits^c = (f, \emptyset)$ where $f(q) = 2$ if $q \in \inits$ and $f(q) = \bot$ otherwise.
    \item $\trans^c$ is then defined as follows:
    \begin{enumerate}
\item if $O \neq \emptyset$, then $\trans^c((f, O), a) = \setcond{ (f', \delta^{e}(O, a) \setminus \odd{f'})}{ f' \cover{\trans^{e}}{a} f } $, 
\item if $O = \emptyset$, then $\trans^c((f, O), a) =  \setcond{ (f', \even{f'} ) }{ f' \cover{\trans^{e}}{a} f } $).
\end{enumerate}
where  $\trans^e$ is the reduced transition function at a level whose corresponding set of states and input letter are $\setcond{q \in \states}{ f(q) \neq \bot}$ and $a$, respectively.
\end{itemize}
\end{customdef}

Recall that the coverage relation between two level ranking functions $f$ and $f'$, parameterized with  $\trans^{e}$, is defined in Definition~\ref{def:coverage-level-rnk}.
Similarly to Definition~\ref{def:edge-relation-e}, to compute $\trans^e(S_{1}, a)$, one has to first compute the minimal set $S_{min}$ of predecessors of $S' = \trans(S, a)$ where $S$ is the domain of $f$, i.e., the set of states at current level and $a$ is the input letter at current level. 
Thus we have $\trans^e(S_{1}, a) = \trans(S_{1}\cap S_{min}, a)$.
Intuitively, for $w \in \infwords$, $\trans^e$ is used to construct a co-deterministic DAG $G^e_{w,\A}$ over $w$ level by level.
By Lemma~\ref{lem:unbw-rdag-nonacc-branch-max-rank}, the maximum rank of $G^e_{w,\A}$ is at most $2$, which is sufficient in Definition~\ref{def:complement-rank-kv-fanbws} for constructing a ranking function to identify whether $G^e_{w,\A}$ is accepting.
Therefore, with Definition~\ref{def:complement-rank-kv-fanbws}, we can construct a complementary NBW $\A^c$ with $2^{\bigO(n)}$ states.

\begin{theorem}[The Language and Size of $\A^c$ for FANBWs]
\label{thm:complexity-kv-fanbw}
Let $\A$ be an FANBW with $n$ states and $\A^c$ the NBW defined in Definition~\ref{def:complement-rank-kv-fanbws}. Then (1) $\lang{\A^c} = \infwords \setminus \lang{\A}$; and (2) $\A^c$ has $2^{\bigO(n)}$ states.
\end{theorem}
\begin{proof}
The proof for claim (2) is trivial and thus omitted here.
By Lemma~\ref{lem:fanbw-run-acc} and definition of ranking functions, co-deterministic DAGs of $\A$ over $w\in\lang{\A}$ will be rejected in $\A^c$, thus $\lang{\A^c} \subseteq \infwords \setminus \lang{\A}$. According to the proof of Lemma~\ref{lem:unbw-rdag-nonacc-branch-max-rank}, there exists a unique ranking function for each rejecting co-deterministic DAG $G^e_{w,\A}$ of $\A$ over $w \notin \lang{\A}$.
This unique ranking function can be constructed in a way similar to the one in Lemma~\ref{lem:nonacc-run-odd-ranks}.
Since \rkc nondeterministically guesses rankings of $G^e_{w,\A}$, 
there must be a guess of such unique ranking function. It follows that $G^{e}_{w, \A}$ must be accepting in $\A^c$, i.e., $\infwords \setminus \lang{\A} \subseteq \lang{\A^c}$. Thus it holds that $\lang{\A^c} = \infwords \setminus \lang{\A}$.
\end{proof}

In~\cite{DBLP:journals/corr/abs-1110-6183}, Fogarty and Vardi proved that complementing reverse deterministic NBWs with \rkc is doable in $2^{\bigO(n)}$ as the non-reduced DAGs $G_{w,\A}$ are already co-deterministic.
This is because that if $\A$ is reverse deterministic, then each vertex $\vertex{q}{l}$ in $G_{w, \A}$ has at most one predecessor, as $q$ has only one $\wordletter{w}{l}$-predecessor.
It follows that $G_{w, \A}$ is co-deterministic. Similarly to Lemma~\ref{lem:unbw-rdag-num-stable}, the number of (accepting) $\omega$-branches in $G_{w, \A}$ is at most $\size{\states}$. According to Definition~\ref{def:fanbw}, reverse deterministic NBWs are a special class of FANBWs, as stated in Corollary~\ref{coro:rdnbw-is-fanbw}.

\begin{corollary}
\label{coro:rdnbw-is-fanbw}
Let $\A$ be a reverse deterministic NBW.
Then $\A$ is also an FANBW.
\end{corollary}
In contrast, an FANBW is not necessarily a reverse deterministic NBW. For instance, the FANBW $\A$ of Figure~\ref{fig:example} is not reverse deterministic since $q_{1}$ has three $b$-predecessors, namely $q_{0}, q_{1}$ and $q_{2}$. We remark that the construction in~\cite{DBLP:journals/corr/abs-1110-6183} just sets the maximum rank to $2$ in Definition~\ref{def:complement-rank-kv} without modifying the transition function $\trans^c$, which turns out to be a special case of our construction according to Corollary~\ref{coro:rdnbw-is-fanbw}.
 
\section{Slice-Based Algorithm}
\label{sec:slice-based}

In Subsection~\ref{ssec:slice-run-dag}, we first recall the \emph{slice-based} complementation construction (\slc) described in~\cite{DBLP:conf/birthday/VardiW08,kahler2008complementation}, adapted using our notations, which produces a complementary NBW $\A^c$ of $\A$ with $\bigO((3n)^n)$ states. 
Then, in Subsection~\ref{ssec:opt-slice-algo}, we show that for FANBWs, this construction can be simplified while yielding a complementary NBW with $\bigO(4^n)$ states.

\subsection{Slice-Based Algorithm for NBWs}
\label{ssec:slice-run-dag}

Let $\A$ be an NBW, and let $w$ be a word.
\slc uses a data structure called \emph{slice} instead of level rankings to encode the set of vertices at the same level in $G_{w,\A}$.
A slice in \cite{DBLP:conf/birthday/VardiW08} is defined as an ordered sequence of disjoint sets of vertices at the same level. 

We now describe \slc from the perspective of building co-deterministic DAGs.
\slc does the following to construct a co-deterministic DAG $G^{s}_{w,\A}$ as it proceeds along the word $w$.
Here the superscript $s$ for \slc is used to distinguish the construction of co-deterministic DAGs $G^{e}_{w,\A}$ in Section~\ref{sec:reduced-run-dag}.
At level $0$, we may obtain at most two vertices of $G^{s}_{w, \A}$: a vertex $\vertex{S_{1}}{0} = \vertex{\inits \setminus \acc}{0} $ and an $\acc$-vertex $\vertex{S_{2}}{0}= \vertex{\inits \cap \acc}{0}$. Recall that $\inits$ and $\acc$ are the set of initial states and the set of accepting states of $\A$, respectively. Here $S_{1}$ and $S_{2}$ are disjoint. A vertex $\vertex{S_{j}}{i}$ is an $\acc$-vertex if $S_{j} \subseteq \acc$, where $j \geq 1$ and $i \geq 0$. The vertices $\vertex{S_{j}}{i}$ on level $i$ in $G^{s}_{w,\A}$ are ordered from left to right by their indices $j$ where $i \geq 0$ and $1 \leq j \leq n$. During the construction, empty sets $S_{j}$ are removed and the indices of remaining sets are reset according to the increasing order of their original indices.

Assume that on level $i$, the sequence of vertices in $G^{s}_{w, \A}$ is $\vertex{S_{1}}{i}, \cdots, \vertex{S_{k_{i}}}{i}$ where $i \geq 0$ and $1 \leq k_{i} \leq n$. We now describe how \slc constructs the vertices on level $i+1$. First, for a set $S_{j}$ where $ 1 \leq j \leq k_{i}$, on reading the letter $\wordletter{w}{i}$, the set of successors of $S_{j}$ is partitioned into (1) a non-$\acc$ set $S'_{2j-1} = \trans(S_{j}, \wordletter{w}{i}) \setminus \acc$, and (2) an $\acc$-set $S'_{2j} = \trans(S_{j}, \wordletter{w}{i}) \cap \acc$, as a possible new $\acc$-vertex.

This gives us a sequence of sets $S'_{1}, S'_{2}, \cdots, S'_{2k_{i}-1}, S'_{2k_{i}}$. Note that there can be some states in $\A$ present in multiple sets $S'_{j}$ where $j \geq 1$. Here we only keep the rightmost occurrence of a state. Intuitively, different runs of $\A$ may merge with each other at some level and we only need to  keep the right most one and cut off others, as they share the same infinite suffix. This operation does not change whether the co-deterministic DAG $G^{s}_{w, \A}$ is accepting, since at least one accepting run of $\A$ remains and will not be cut off. Formally, for each set $S'_{j}$ where $1 \leq j \leq 2k_{i} $, we define a set $S''_{j} = S'_{j} \setminus \bigcup_{ j<p \leq 2k_{i}} S'_{p}$. This yields a sequence of disjoint sets $S''_{1}, S''_{2}, \cdots, S''_{2k_{i}-1}, S''_{2k_{i}}$.
After removing the empty sets in this sequence and reassigning the index of each set according to their positions, we finally obtain the sequence of sets of vertices on level $i+1$, denoted by $\vertex{S_{1}}{l+1}, \cdots, \vertex{S_{k_{i+1}}}{l+1}$.
Obviously, the resulting sets at the same level are again pairwise disjoint.

Therefore, we define a co-deterministic DAG $G^{s}_{w, \A} = (V, E)$ of $\A$ over $w$ for an NBW $\A$ as follows:
\begin{itemize}
\item Vertices. $V = \bigcup_{l\geq 0, 1 \leq j \leq k_{i}} \setnocond{\vertex{S_{j}}{l}}$.
\item Edges. There is an edge from $\vertex{S_{j}}{l}$ to $\vertex{S_{h}}{l+1}$ iff $S_{h}$ is either $S''_{2j-1}$ or $S''_{2j}$ as defined above where $1 \leq j \leq k_{i}$ and $1\leq h \leq k_{i+1}$. 
\end{itemize}

By the definition of $G^{s}_{w, \A}$, each vertex $\vertex{S_{h}}{l+1}$ in which $S_{h}$ is either $S''_{2j-1}$ or $S''_{2j}$ computed from $S_{j}$ has at most one predecessor $\vertex{S_{j}}{l}$.
Thus $G^{s}_{w,\A}$ is co-deterministic.
Similarly, we have the following Lemma \ref{lem:finite-ambiguity-stable-levels}.

\begin{lemma}[Co-Deterministic DAGs for NBWs~\cite{DBLP:conf/birthday/VardiW08}]
\label{lem:finite-ambiguity-stable-levels}
Let $w \in \infwords$ and $G^{s}_{w, \A}$ be the co-deterministic DAG as defined above.
Then (1) the number of (accepting) $\omega$-branches in $G^{s}_{w, \A}$ is at most the number of states in $\A$. (2) $w$ is accepted by $\A$ if and only if $G^{s}_{w,\A}$ is accepting.
(3) There exists a stable level $l \geq 1$ in $G^{s}_{w,\A}$ such that all $\acc$-vertices after level $l$ are finite if and only if $w \notin \lang{\A}$.
\end{lemma}

\slc for general NBWs can be viewed as consisting of two components: (1) based on the construction of co-deterministic DAGs $G^{s}_{w, \A}$ over $w$ above,  NBWs can be translated to FANBWs~\cite{LodingP18} and (2) a specialized complementation algorithm for FANBWs.
In~\cite{DBLP:conf/birthday/VardiW08}, $\slc$ utilizes these two components at the same time for computing the complementary NBW $\A^c$.

A state of $\A^c$ is an ordered sequence of tuples $(S_{1}, l_{1}), \cdots, (S_{h}, l_{h})$ where the ordered sequence
$(S_{1}, \cdots, S_{h}) $ is a slice,
and each vertex $\vertex{S_j}{l}$ is decorated with a label  $l_{j} \in \setnocond{\textsf{die}, \textsf{inf}, \textsf{new}}$.
The level index $l$ is omitted during the construction of $\A^c$.
Intuitively,
\begin{itemize}
    \item  $\textsf{die}$-labelled vertex means that those states in $S_{j}$ are currently being inspected.
    For $w$ to be accepted (i.e., $w\not\in\mathcal{L}(\mathcal{A})$),  $\textsf{die}$-labelled vertices should eventually reach empty set after a finitely many steps, thus become finite.
    Recall that empty sets will be removed in the construction of $G^s_{w,\A}$.
    \item $\textsf{inf}$-labelled vertex indicates all states that never reach accepting states.
    \item $\textsf{new}$-labelled vertex records new encountered states, that should be inspected later once the $\textsf{die}$-labelled vertex becomes empty. 
\end{itemize}

Obviously, here $h$ is at most the number $n$ of states in $\A$. 
While for FANBWs, thanks to their finite ambiguity, the construction for co-deterministic DAGs can be simplified (see Section~\ref{sec:reduced-run-dag}): we can even use three components $(N,C,B)$ to compactly encode the slice and their labels. 
We  postpone the details of the construction to the next subsection. Now we recall the complexity of the above slice based construction:
\begin{lemma}[The Language and Size of $\A^c$ for NBWs~\cite{DBLP:conf/birthday/VardiW08} ]
\label{lem:size-language-slice}
Let $\A$ be an NBW with $n$ states and $\A^c$ the NBW constructed by $\slc$ in Section~\ref{sec:slice-based}.
Then $\lang{\A^c} = \infwords \setminus \lang{\A}$ and $\A^c$ has $\bigO((3n)^n)$ states.
\end{lemma}

\subsection{Slice-Based Algorithm for FANBWs}
\label{ssec:opt-slice-algo}

We now propose the specialized complementation construction for FANBWs. Recall that, as discussed in Subsection~\ref{ssec:slice-run-dag}, this construction is also the second component of \slc, used for complementing general NBWs.

We first provide some intuitions. According to Lemma~\ref{lem:unbw-rdag-nonacc-branch-stable},  given a word $w \notin \lang{\A}$, there exists a stable level $k$ in the co-deterministic DAG $G^{e}_{w,\A}$ such that each $\acc$-vertex on a level after $k$ is finite. Therefore, in the construction of $\A^c$, we can nondeterministically guess level $k$ and then use breakpoint construction to verify that our guess is correct, in analogy with \rkc. More precisely, when constructing the complementary NBW $\A^c$, there are the \emph{initial phase} and the \emph{accepting phase}. The initial phase is purely a subset construction to trace the reachable states of each level of the co-deterministic DAG $G^{e}_{w,\A}$ over $w$. On reading a letter at a state of $\A^c$ (called \emph{macrostate}) in the initial phase, the run of $\A^c$ over $w$ (called \emph{macrorun}) either continues to stay in the initial phase or jumps to the accepting phase. Once entering the accepting phase, we guess that the macrorun of $\A^c$, which consists of multiple runs of $\A$, has reached the stable level $k$. Thus in the accepting phase, we need a breakpoint construction to verify that the guess is correct, i.e., that all $\acc$-vertices after level $k$ are finite.

In the accepting phase, we use a macrostate, represented as a triple $(N, C, B)$, to encode the set of vertices and their labels on a level after $k$ in the co-deterministic DAG $G^e_{w,\A}$ (or $G^{s}_{w,\A}$ for general NBWs accordingly), where
\begin{itemize}
\item the set $N$ keeps all the reachable vertices on the level, corresponding to the set of all vertices labelled with $\textsf{die}$, $\textsf{inf}$ and $\textsf{new}$;
\item the set $C$ keeps all the finite vertices on the level. That means, it contains both $\textsf{new}$-labelled vertices recording new encountered states, and $\textsf{die}$-labelled vertices being inspected now.
\item the set $B \subseteq C$ as a breakpoint construction is used to verify that the guess on the set $C$ of finite vertices is correct, corresponding to the set of vertices labelled with $\textsf{die}$.
\end{itemize} 
Recall that $\textsf{die}$, $\textsf{inf}$ and $\textsf{new}$ are three labels of vertices used in \slc for complementing general NBWs, as described in Subsection~\ref{ssec:slice-run-dag}.
The specialized complementation algorithm for FANBWs is formalized below.
\begin{definition}
\label{def:unbw-complement}
Let $\A = (\states, \inits, \trans, \acc)$ be an FANBW.
We then define an NBW $\A^c= (\states^c, \inits^c, \trans^c, \acc^c)$ as follows.
\begin{itemize}
\item $\states^c \subseteq 2^{\states}\cup 2^{\states}\times 2^{\states} \times 2^{\states}$;
\item $\inits^c = \setnocond{\inits }$;
\item $\trans^c = {\ntrans}^c \cup {\jtrans}^c\cup{\dtrans}^c$ is defined as follows:
\begin{enumerate}
\item ${\ntrans}^c(S, a) = \trans^{e}(S, a)$ for $S \subseteq \states $ and $a\in\alphabet$ where $\trans^{e}$ is the reduced transition function at current level whose corresponding set of states and input letter are $S$ and $a$, respectively. (intuition: subset construction to organize the macrorun before the guess point).
\item ${\jtrans}^c(S, a) = {\dtrans}^c((N, C ,B), a)$ where $N = S, B = S\cap\acc $ and $C = B$ (intuition: make the guess point to be the macrostate $(N, C, B)$).
\item ${\dtrans}^c((N, C, B), a) = (N', C', B')$ where $\trans^{e}$ is the reduced transition function at current level whose corresponding set of states and input letter are $N$ and $a$, respectively, and 
\begin{itemize}
\item $N' = \trans^{e}(N, a)$ (intuition: tracing the reachable states correctly),
\item $C' = \trans^{e}(C, a) \cup (N' \cap \acc)$ (intuition: tracing the runs which has visited accepting states after the guess point), and
\item if $B \neq \emptyset$, then $B' = \trans^{e}(B, a)$ and otherwise $B' = C'$ (intuition: $B = \emptyset$ means all runs which have visited accepting states are finite and $B \neq \emptyset$ indicates that previous runs are still under inspection).
\end{itemize}
\end{enumerate}
\item $\acc^c = \setcond{(N, C, B) \in \states^c}{ B = \emptyset }$.
\end{itemize}
\end{definition}

\begin{remark}

As a side remark, we note that the complementary NBW constructed by Definition~\ref{def:unbw-complement} is \emph{limit deterministic}, as the state set $\states^c$ of $\A^c$ can be partitioned into two disjoint sets $\states^c_{N} \subseteq 2^{\states}$ and $\states^c_{D} \subseteq 2^{\states} \times 2^{\states} \times 2^{\states}$ such that
1) $\acc^c \subseteq \states^c_{D}$ and 2) for each state $q \in \states^c_{D}$ and $a \in \alphabet$, we have that $\size{\trans^c(q, a)} \leq 1$.   
\end{remark}

\begin{theorem}[The Language and Size of $\A^c$ for FANBWs]
Let $\A$ be an FANBW with $n$ states and $\A^{c}$ be the NBW defined by Definition~\ref{def:unbw-complement}.
Then
(1) $\lang{\A^c} = \infwords \setminus \lang{\A}$; and
(2) $\A^c$ has $2^n + 4^n$ states. 
\end{theorem}
\begin{proof}
We prove claim (1) as follows. 
Suppose $w \in \lang{\A} $, our goal is to prove $w$ is not accepted by $\A^c$.
Assume that the corresponding accepting run of $\A$ over $w$ is $\run$ and $\run'$ is a macrorun of $\A^c$ over $w$.
Then for the macrorun $\run'$:
(1) if $\run'$ only visits states of the form $s \in 2^{\states}$, then $\run'$ is not accepted by $\A^c$ since no accepting $\A^{c}$-states will be visited;
(2) if $\run'$ is a macrorun of the form $s_0, \cdots s_{k-1}, (N_k, C_k, B_k) (N_{k+1}, C_{k+ 1}, B_{k+1}) \cdots$, $\run$ will visit some accepting state, say $q_f \in \acc$ infinitely often.
Then at some point, say in state $(N_j, C_j, B_j)$, we have $q_f \in B_j$ or $q_f \in C_j$.
If $q_f \in B_j$, then for every $p \geq j$, we have $B_{p} \neq \emptyset$ according to Lemma~\ref{lem:vertices-one-pred};
otherwise $q_f \in C_j$, then either at some point, say $p > j$, $q_f$ will be moved to $B_p$ when $B_{p-1} = \emptyset$, or $q_f \in C_{p}$ for each $p \geq j$, which indicates that $B_{p} \neq \emptyset$ for $p \geq j$.
Therefore, $w$ is not accepted by $\A^c$.  

Assume that $w \notin \lang{\A}$, our goal is to prove that there exists an accepting macrorun $\run'$ of $\A^c$ over $w$.
The proof idea is to analyze the co-deterministic DAG $G^{e}_{w,\A}$ of $\A$ over $w$.
According to Lemma~\ref{lem:unbw-rdag-nonacc-branch-stable}, there exists a stable level $k \geq 1$ such that every $\acc$-vertex on a level after $k$ of $G^e_{w,\A}$ is finite.
Therefore, the set $B$ on $\run'$ will become empty infinitely often, i.e., $w$ is accepted by $\A^c$.

We now prove claim (2).
By Definition~\ref{def:unbw-complement}, the number of possible states of the form $s \in 2^{\states}$ is $2^n$.
For each state $p = (N, C, B) \in \states^c$ of $\A^c$, we have that $C \subseteq N$ and $B \subseteq C$.
Then for a state $q \in \states$: 
(i) it will either be absent or present in $N$;
(ii) for a state $q \in N$, one of the following three possibilities holds:
$q$ is only in $N$, $q$ is both in $C$ and $N$ and $q$ is both in $B$ and $C$.
Therefore $\A^c$ has at most $2^n + 4^n$ states.
\end{proof}

As a consequence of Definition~\ref{def:unbw-complement}, we can define a subsumption relation between the macrostates of $\A^c$ below.
\begin{corollary}[Subsumption Relation between Macrostates]
\label{coro:inclusion}
Let $\A$ be an FANBW and $\A^{c}$ the complementary NBW of $\A$ defined by Definition~\ref{def:unbw-complement}, and $m = (N, C, B)$ and $m' = (N', C', B')$ are two macrostates of $\A^{c}$ such that $N = N'$ and $C \subseteq C'$. Then $\lang{(\A^{c})^{m'}} \subseteq \lang{(\A^{c})^{m}} $ or $m$ subsumes $m'$.
\end{corollary}
\begin{proof}
Let $w = a_{0} a_{1} \cdots \in \infwords$.
Let $\rho = (N_{0} = N, C_{0} = C, B_{0} = B) (N_{1}, C_{1}, B_{1}) \cdots (N_{k}, C_{k}, B_{k}) \cdots$ be the macrorun of $(\A^c)^m$ over $w$.
Similarly, the macrorun of $(\A^c)^{m'}$ over $w$ is $\rho' = (N'_{0} = N', C'_{0} = C', B'_{0} = B') (N'_{1}, C'_{1}, B'_{1}) \cdots (N'_{k}, C'_{k}, B'_{k}) \cdots$.
Assume that $w \in \lang{(\A^{c})^{m'}}$, i.e., there are infinitely many empty $B'$-sets in $\rho'$ according to Definition~\ref{def:unbw-complement}.
It follows that the level $0$ in the co-deterministic DAG $G^e_{w, \A^{N'}}$ of $\A^{N'}$ over $w$ is a stable level, i.e., each $\acc$-vertex in $G^e_{w, {\A^{N'}}}$ is finite.
(Recall that $\A^{N'}$ is an NBW obtained from $\A$ by setting the set of initial states of $\A$ to $N'$.)
This is because that by Definition~\ref{def:unbw-complement}, each branch from an $\acc$-vertex in $G^e_{w, {\A^{N'}}}$ will eventually be put in the $B'$-set and if one such branch is not finite, the $B'$-set will become empty for only finitely many times, contradicting with the assumption that $w \in \lang{(\A^{c})^{m'}}$.
By definition of the construction of co-deterministic DAGs in Section~\ref{sec:reduced-run-dag}, the co-deterministic DAG $G^e_{w, \A^{N}}$ of $\A^{N}$ over $w$ is identical to $G^e_{w, {\A^{N'}}}$ since $N = N'$.
Consequently, the level $0$ is also a stable level in $G^e_{w, {\A^{N}}}$.
That is, each $\acc$-vertex in $G^e_{w, {\A^{N}}}$ is also finite. 
Since the $B'$-set in $\rho'$ becomes empty and is reset to $C'$ for infinitely many times, all branches from $C'$ are finite.
It follows that all the branches from $B \subseteq C$ are also finite since $C \subseteq C'$.
Then there exists a least integer $j \geq 0$ in $\rho$ such that $B_{j} = \emptyset$.
Since all branches in the $C$-set (including new branches coming from the $N$-set) are finite, there are infinitely many integers $k \geq j$ such that $B_{k} = \emptyset$ in $\rho$.
It follows that $w \in \lang{(\A^{c})^{m}}$, which indicates that $\lang{(\A^{c})^{m'}} \subseteq \lang{(\A^{c})^{m}} $.
\end{proof}

Corollary~\ref{coro:inclusion} provides the possibility to avoid the exploration of $m'$ when $\lang{(\A^c)^{m}}$ has already been found to be empty, when checking the language-containment between an NBW and an FANBW $\A$.
It follows that one can also use this subsumption relation to avoid construction of redundant macrostates during the construction of $\A^c$, thus reducing the number of macrostates in $\A^c$. 

\section{Conclusion and Future Work}\label{sec:conclusion}
This work exploits co-deterministic DAGs over infinite words as a unified tool to optimize both \rkc and \slc constructions.
Consequently, we have improved the complexity of the classical \rkc and \slc constructions for FANBWs, respectively, to $2^{\bigO(n)}$ from $2^{\bigO( n \log n)}$ and to $\bigO( 4^n)$ from $\bigO( (3n)^n)$, based on co-deterministic DAGs.
As a further contribution, we view the \slc algorithm explicitly as the construction of co-deterministic DAGs and a specialized complementation algorithm for FANBWs. We then provide a subsumption relation between states in the complementary NBWs of FANBWs in hope of improving the containment checking between an NBW and an (FA)NBW.

As future work, we plan to study whether $\bigO(4^n)$ is also the lower bound for the complementation of FANBWs.
An empirical evaluation on how the subsumption relation between macrostates proposed in Corollary \ref{coro:inclusion} will benefit the containment checking problem is worthy of exploring.
Moreover, we will also explore a Ramsey-based complementation construction based on co-deterministic DAGs. Another line of future work is studying determinization constructions for FANBWs.
Finally, it is possible to use our work to improve the program-termination checking framework proposed in \cite{DBLP:conf/cav/HeizmannHP14} if one generalizes a terminating path to an FANBW.

\paragraph{Acknowledgment}
We thank Rachel Faran, Yih-Kuen Tsay and anonymous reviewers for their valuable inputs at different stages to this project.
This work is partially supported by Key-Area Research and Development Program of Guangdong Province (grant no. 2018B010107004), the National Natural Science Foundation of China (grant nos.\ 61761136011, 61532019), NSF grants IIS-1527668, CCF-1704883, IIS-1830549, and an award from the Maryland Procurement Office.

\nocite{*}
\bibliographystyle{eptcs}
\bibliography{generic}

\end{document}